\documentclass[a4paper, 10pt, twocolumn]{article}

\usepackage{amsmath,amsfonts,amssymb,amsthm,commath,graphicx, epstopdf, lmodern, txfonts}
\usepackage{graphicx, caption, subcaption}
\usepackage[latin1]{inputenc}
\usepackage[english]{babel}
\usepackage[colorlinks=true]{hyperref}
\usepackage{xspace}

\newcommand{\dotex}{\frac{d}{dt}}
\newcommand{\tr}[1]{\text{Tr}\left(#1\right)}

\newcommand{\ket}[1]{\ensuremath{|#1\rangle}\xspace}
\newcommand{\bra}[1]{\ensuremath{\langle #1|}\xspace}
\newcommand{\braket}[2]{\langle #1|#2 \rangle}

\newcommand{\ba}{\boldsymbol{a}}

\newcommand{\bb}{\boldsymbol{b}}

\newcommand{\bH}{\boldsymbol{H}}

\newcommand{\bN}{{\boldsymbol{N}}}

\newcommand{\fL}{\mathfrak{L}}

\newcommand{\Dp}{\mathfrak{D}_{\alpha}}
\newcommand{\Dm}{\mathfrak{D}_{\alpha}^\dagger}

\newcommand{\rkz}{K_0^s}

\newtheorem{lemma}{Lemma}

\title{\LARGE \bf
Adiabatic elimination for open quantum systems with effective Lindblad master equations \thanks{ This work was partially supported by the Projet Blanc ANR-2011-BS01-017-01 EMAQS.}
}

\author{ R. Azouit\thanks{Centre Automatique et  Syst\`{e}mes, Mines-ParisTech, PSL Research University.
60 Bd Saint-Michel, 75006 Paris, France. },  
A. Sarlette\thanks{ INRIA Paris, 2 rue Simone Iff, 75012 Paris, France; and Ghent University / Data Science Lab, Technologiepark 914, 9052 Zwijnaarde, Belgium.},   P. Rouchon\footnotemark[2]
}

\begin{document}

\maketitle
\thispagestyle{empty}
\pagestyle{empty}

%%%%%%%%%%%%%%%%%%%%%%%%%%%%%%%%%%%%%%%%%%%%%%%%%%%%%%%%%%%%%%%%%%%%%%%%%%%%%%%%
\begin{abstract}

We consider an open quantum system described by a  Lindblad-type master equation with two times-scales. The fast time-scale is strongly dissipative   and  drives the system towards a low-dimensional decoherence-free space. To perform the  adiabatic elimination of this fast relaxation, we propose  a geometric asymptotic expansion based on  the  small positive parameter describing the time-scale  separation. This expansion exploits   geometric singular perturbation theory and center-manifold techniques.   We conjecture that, at any order,  it provides  an effective  slow Lindblad master equation and  a  completely positive parameterization of the slow invariant sub-manifold associated to the low-dimensional decoherence-free space.  By preserving  complete positivity and trace, two important structural properties attached to   open quantum dynamics, we obtain a reduced-order model  that directly conveys a physical interpretation since it relies  on  effective Lindbladian  descriptions of the slow evolution. At  the first order, we derive simple formulae for the effective Lindblad master equation.   For a specific type  of fast dissipation, we show how any  Hamiltonian perturbation  yields  Lindbladian second-order corrections to the first-order slow evolution governed by the  Zeno-Hamiltonian. These results are illustrated on a  composite system made of a strongly dissipative harmonic oscillator, the ancilla, weakly coupled to another quantum system.

\end{abstract}

%%%%%%%%%%%%%%%%%%%%%%%%%%%%%%%%%%%%%%%%%%%%%%%%%%%%%%%%%%%%%%%%%%%%%%%%%%%%%%%%
\section{Introduction}

Solving the equation of evolution for a open quantum system - the Lindblad master equation \cite{BreuerPetruccioneBook} - is generally tedious.
To gain better physical insight and/or for numerical simulations, it is of wide interest to compute rigorous reduced models of quantum dynamical systems. In a typical case, a system of interest is coupled to an ancillary system expressing a measurement device or a perturbing environment \cite{haroche-raimondBook06}. The quantum dynamics describes the joint evolution of both systems and in order to focus only on the system of interest we want to determine a dynamical equation for the system of interest only, from which we have ``eliminated'' the ancillary system.

%C2: adiabatic elimination: from Hamilton to Lindblad, history
A standard tool for model reduction is to use the different timescales of the complete system to separate the quantum dynamics into fast and slow variables and then eliminate the fast ones. This technique is known as adiabatic elimination. In quantum Hamiltonian systems, regular perturbation theory can be easily applied as the propagator remains unitary, and the construction of the reduced model to various orders of approximation is standard \cite{sakurai2011modern}. In contrast, for open quantum systems, described by a Lindblad master equation \cite{BreuerPetruccioneBook}, the case is much more complicated and involves singular perturbation theory.
Several particular examples have been treated separately. In \cite{BrionJPA07} different methods are proposed to perform an adiabatic elimination up to second-order on a lambda system. In \cite{mirrahimi-rouchonieee09} and \cite{ReiteS2012PRA} the problem of excited states decaying towards $n$ grounds state is treated. A specific atom-optics dynamics is investigated in \cite{AtkinsPRA03}. In the presence of continuous measurement, \cite{CernotikPRA15} presents a method of adiabatic elimination for systems with Gaussian dynamics.

%C3: general theory, introduce structure
However, no general and systematic method has been developed yet for adiabatic elimination in systems with Lindblad dynamics. More precisely, treating the Lindblad master equation as a usual linear system, or applying the Schrieffer-Wolff formalism which is generalized in \cite{Kessl2012PRA} to Lindblad dynamics, requires the inversion of super-operators which can be troublesome both numerically and towards physical interpretation. In \cite{AzouitCDC15} we made a first attempt to circumvent this inversion, using invariants of the dynamics. This provides a first-order expansion only, and in general linear form --- i.e.~not necessarily with the structure of a Lindblad equation.

% C: our contributions
In the present paper, we propose a geometric  method to perform an adiabatic elimination for open quantum systems, whose key feature is that \emph{the resulting reduced model is explicitly described by an effective  Lindblad equation. The reduced system is parameterized by a reduced density operator and the mapping from the reduced model to the initial system state space is expressed in terms of Kraus operators, ensuring a trace-preserving completely positive map.} By preserving these structural properties of open quantum dynamics, we obtain a reduced model that directly conveys a physical interpretation. As far as we know, combining asymptotic expansion with  completely positive map and Lindbladian formulation  has  never been addressed before.  This work  is a first attempt to investigate  the interest of such combination with  lemmas~\ref{lem:Ls1}, \ref{lem:K1} and~\ref{lem:Fs2} underlying the conjecture illustrated on figure~\ref{fig:GeomAdiabElimLindblad}.

Our method applies to general open quantum systems with two timescales, described by two general Lindbladian super-operators \eqref{eq:DynEpsilon}, and where the fast Lindbladian makes the system converge to a decoherence-free subspace of the overall Hilbert space. We then use a geometric approach based on center manifold techniques \cite{carr-book} and geometric singular perturbation theory \cite{Fenichel79} to obtain an expansion of the effect of the perturbation introduced by the slow Lindbladian on this decoherence-free subspace. For general Lindbladians satisfying this setting, we get explicit formulas for the Lindblad operators describing the first-order expansion. In the particular case of a Hamiltonian perturbation, we retrieve the well known Zeno effect. Furthermore, for a fast Lindbladian described by a single decoherence operator and subject to a Hamiltonian perturbation, we derive explicit formulas for the first-order effect on the location of the center manifold and for Lindblad operators describing the second-order expansion of the dynamics. This allows to highlight how a first-order Zeno effect is associated to second-order decoherence.

We apply our method to a quantum system coupled to a highly dissipative quantum harmonic oscillator (ancilla). Our general formulas directly provide an effective  Lindblad master equation of the reduced model where this ancilla is eliminated. The result for this example is well known \cite{CarmichaelBook07}, which allows us to emphasize how the correct results are obtained also on infinite-dimensional systems, and to appreciate the computational simplicity of our method in comparison with previous ones.

The paper is organized as follows. Section \ref{sec:system} presents the structure of two-timescales master equation for open quantum systems, as well as the assumptions and properties of the unperturbed system used for deriving our results. In section \ref{sec:firstorder} we present a geometric approach for performing the adiabatic elimination and derive a first-order reduced model for arbitrary perturbations. In section \ref{sec:secondorder} we develop the second-order expansion for a class of systems. In section \ref{sec:example} we illustrate the method where the ancilla  is a highly dissipative harmonic oscillator.

\section{A class of perturbed master equations}\label{sec:system}

Denote by $\mathcal{H}$ a Hilbert space of finite dimension, by $\mathcal{D}$ the compact convex set of density operators  $\rho$ on $\mathcal{H}$ ($\rho$ is Hermitian, nonnegative and trace one). We consider a two-time scale dynamics on $\mathcal{D}$ described by the following master differential equation
\begin{equation}\label{eq:DynEpsilon}
 \dotex \rho = \fL_0(\rho) + \epsilon \fL_1(\rho)
\end{equation}
where $\epsilon$ is a small positive parameter and the linear super-operators $\fL_0$ and $\fL_1$ are of Lindbladian forms~\cite{BreuerPetruccioneBook}. That is, there exist  two finite families  of operators on $\mathcal{H}$, denoted by $(L_{0,\nu})$ and $(L_{1,\nu})$,  and two Hermitian operators $H_0$ and $H_1$ (called Hamiltonians) such that, for $r=0,1$, we have
\begin{equation}\label{eq:GenLindblad}
\fL_r(\rho)= -i[H_r,\rho] + \sum_{\nu} L_{r,\nu}\rho L_{r,\nu}^\dag - \tfrac{1}{2} L_{r,\nu}^\dag   L_{r,\nu}\rho -\tfrac{1}{2} \rho L_{r,\nu}^\dag   L_{s,\nu}
.
\end{equation}
It is well known that if the initial condition $\rho(0)$ belongs to $\mathcal{D}$, then the solution $\rho(t)$ of~\eqref{eq:DynEpsilon} remains in $\mathcal{D}$ and is defined for all $t\geq 0$. It is also a known fact that the flow of~\eqref{eq:DynEpsilon} is a contraction for many distances, such as the one derived from the nuclear norm $\|\cdot\|_1$: for any trajectories $\rho_1$ and $\rho_2$ of~\eqref{eq:DynEpsilon}, we have $\| \rho_1(t) -\rho_2(t) \|_1 \leq \| \rho_1(t') -\rho_2(t') \|_1$ for all $t\geq t'$; see e.g.~\cite[Th.9.2]{nielsen-chang-book}.

 We assume that,  for $\epsilon=0$, the unperturbed master equation  $\dotex \rho =\fL_0(\rho)$ converges to a  stationary regime. More precisely, we assume that the unperturbed master equation admits  a sub-manifold of stationary operators  coinciding with the $\Omega$ limit set of  its trajectories. Denote by
$$
\mathcal{D}_0=\left\{\rho\in\mathcal{D}~\big|~\fL_0(\rho)=0\right\}
$$
this stationary manifold:  it  is  compact and convex. We thus  assume that for all $\rho_0\in\mathcal{D}$, the solution of $\dotex \rho=\fL_0(\rho)$ with $\rho(0)=\rho_0$ converges for $t$ tending to $+\infty$ towards an element of $\mathcal{D}_0$ denoted by
$R(\rho_0)= \lim_{t\rightarrow +\infty} \rho(t)$.
Since for any $t\geq 0$, the propagator $ e^{t\fL_0}$ is a completely  positive linear map~\cite[Chap.8]{nielsen-chang-book}, $R$ is also a completely positive map. By Choi's theorem~\cite{ChoiLAA1975} there exists a finite set of operators on $\mathcal{H}$ denoted by $(M_\mu)$ such that
\begin{equation}\label{eq:Rkraus}
   R(\rho_0)= \sum_{\mu} M_\mu \rho_0 M^\dag_\mu
\end{equation}
with $\sum_{\mu} M_\mu^\dag M_\mu = I$,  the identity operator on $\mathcal{H}$. The form \eqref{eq:Rkraus}
is called a Kraus map.
We thus assume that
$$
\mathcal{D}_0= \left\{ R(\rho)~\big|~\rho \in\mathcal{D} \right\}\text{ and }
\forall \rho\in\mathcal{D}_0, ~R(\rho)=\rho.
$$

An invariant operator  attached to the dynamics  $\dotex \rho = \fL_0(\rho)$ is an Hermitian operator $J$ such that for any time $t \geq 0$ and any initial state $\rho_0=\rho(0)$, we have
$\tr{J \rho(t)} = \tr{J \rho_0}$. Such invariant operators $J$ are characterized by the fact that $\fL_0^*(J)=0$ where the adjoint map to $\fL_0$ is given by
$$
\fL^*_0(A)= i[H_0,A] +\sum_{\nu} L_{0,\nu}^\dag A L_{0,\nu} - \tfrac{1}{2} L_{0,\nu}^\dag   L_{0,\nu}A -\tfrac{1}{2} A L_{0,\nu}^\dag   L_{0,\nu}
$$
for any Hermitian operator $A$.

Thus by taking the limit for $t$ tending to $+\infty$ in $\tr{J \rho(t)} = \tr{J \rho_0}$,  we have,  for all Hermitian operators $\rho_0$,
$\tr{J R(\rho_0)} = \tr{J \rho_0}$. Denote by $R^*$ the adjoint map associated to $R$:
\begin{equation}\label{eq:Radjoint}
  R^*(A)= \sum_{\mu} M_{\mu}^\dag A M_\mu
\end{equation}
for any Hermitian operator $A$ on $\mathcal{H}$. Then, $\tr{R^*(J)\rho_0 } = \tr{J \rho_0}$ for any  Hermitian operator $\rho_0$ is equivalent to
$R^*(J)= J $. I.e.~invariant operators $J$ are characterized by $\fL_0^*(J)=0$ and  satisfy  $R^*(J)=J$.

We assume additionally that $\mathcal{D}_0$ coincides with the set of density operators  with support in $\mathcal{H}_0$, a subspace of $\mathcal{H}$. In other words the unperturbed master equation features a decoherence-free space $\mathcal{H}_0$.
Denote by $P_0$ the operator on $\mathcal{H}$ corresponding to orthogonal projection onto $\mathcal{H}_0$. Consequently,  for any Hermitian operator $\rho$, we have $P_0 R(\rho)=R(\rho) P_0=R(\rho)$. Thus $\tr{R^*(P_0) \rho}=\tr{R(\rho)}=\tr{\rho}$ for all $\rho$ which implies:
\begin{equation}\label{eq:RP0}
R^{*}(P_0)=I
.
\end{equation}
Moreover, for any  vector   $\ket{c}$ in $\mathcal{H}_0$,   $R(\ket{c}\bra{c})= \ket{c}\bra{c}$. This implies that, for the Kraus map~\eqref{eq:Rkraus}, there exists a family of complex numbers $\lambda_\mu$ such that  $\sum_{\mu} |\lambda_\mu|^2=1$ and
\begin{equation}\label{eq:MmuC}
  \forall \ket{c}\in\mathcal{H}_0,~M_\mu \ket{c} = \lambda_\mu \ket{c}
  .
\end{equation}

%%%  First order  %%%

 \section{First order expansion for arbitrary perturbations} \label{sec:firstorder}

We consider here the perturbed master equation~\eqref{eq:DynEpsilon} whose unperturbed part $\dotex \rho = \fL_0(\rho)$ satisfies  the assumptions of Section \ref{sec:system}:   any trajectory converges to a steady-state; the set of steady-states $\mathcal{D}_0$ coincides with the set of density operators with support on a subspace $\mathcal{H}_0$ of $\mathcal{H}$. This section develops a first-order expansion versus $\epsilon$ of~\eqref{eq:DynEpsilon} around $\mathcal{D}_0$.

Denote by $\mathcal{H}_s$ (subscript $s$ for slow), an abstract Hilbert space with the same dimension as
$\mathcal{H}_0$. Denote by $\mathcal{D}_s$ the set of density operators on $\mathcal{H}_s$. Denote by $\{\ket{\nu}\}$ (resp. $\{\ket{c_\nu}\}$) a Hilbert basis of $\mathcal{H}_s$ (resp. $\mathcal{H}_0$). Consider the Kraus map $K_0$ defined by
\begin{equation}\label{eq:K0}
  \forall \rho_s\in\mathcal{D}_s,~K_0(\rho_s) = S_0 \rho_s S_0^\dag \in\mathcal{D}
\end{equation}
 where $S_0=\sum_\nu \ket{c_\nu} \bra{\nu}$.
 We have $S_0 S_0^\dag =P_0$, the orthogonal projector onto $\mathcal{H}_0$ and $S_0^\dag S_0=I_s$, the identity operator on $\mathcal{H}_s$.

\begin{figure}[h]
  \centering
  \includegraphics[width=0.45\textwidth]{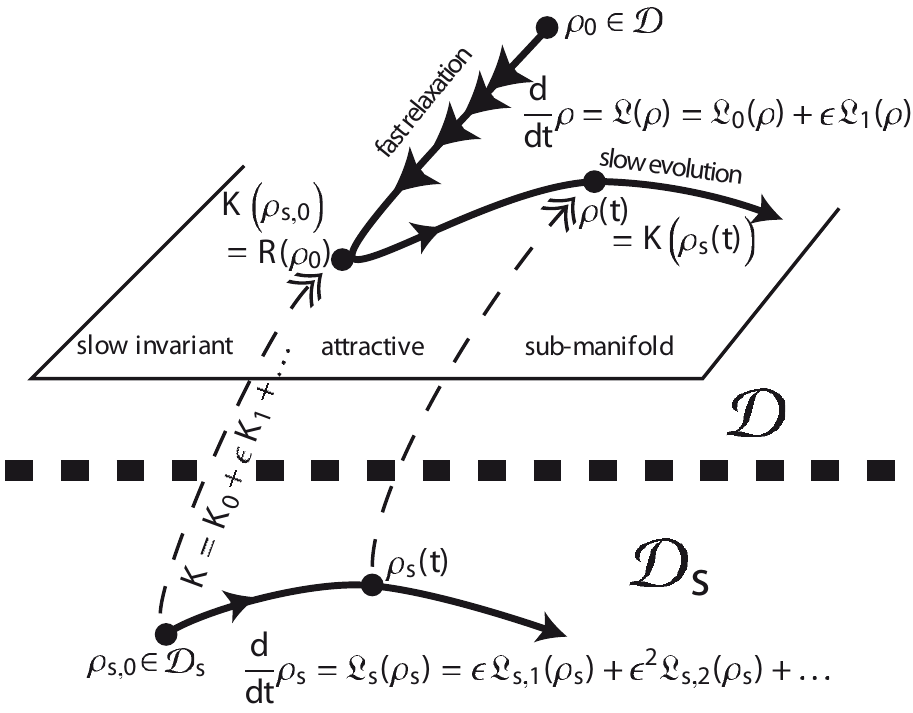}\\
  \caption{\small Adiabatic elimination, based on geometric singular perturbation theory, of the  fast relaxation dynamics $\dotex \rho = \fL_0(\rho)$  for an open-quantum system governed by  the Lindbladian master equation  $\dotex \rho = \fL(\rho)= \fL_0(\rho) + \epsilon \fL_1(\rho)$ where $\epsilon$ is a small positive parameter. It provides  two asymptotic expansions  of the  slow dynamics. The  parameterization  via density operators $\rho_s$ of the slow invariant attractive sub-manifold  (related to  $\fL_0(\rho)=0$)  is based on  the   map $\rho=K(\rho_s)=K_0(\rho_s)+\epsilon K_1(\rho_s) + \ldots$.  The slow dynamics corresponds to $\dotex \rho_s = \fL_s(\rho_s)= \epsilon \fL_{s,1}(\rho_s)+ \epsilon^2 \fL_{s,2}(\rho_s)+\ldots$. The super-operators $K_0$, $K_1$, \ldots  and  $\fL_{s,1}$,  $\fL_{s,2}$, \ldots are obtained  by identifying  terms of identical order versus $\epsilon$  in the geometric invariance condition\\
  $
    \fL_0\big(K_0+\epsilon K_1\big)+ \epsilon \fL_1\big(K_0+\epsilon K_1+\ldots\big)$
    \\
    $\hspace*{2.em} = K_0\big(\epsilon \fL_{s,1} + \epsilon^2 \fL_{s,2}+\ldots\big)
   + \epsilon K_1 \big(\epsilon \fL_{s,1} + \epsilon^2 \fL_{s,2}+\ldots\big)  + \ldots~.
  $\\
   We conjecture that,  at any order versus $\epsilon$,  the super-operator $K$ is a Kraus map (up to higher-order corrections) and the slow evolution $ \dotex \rho_s=\fL_s(\rho_s)$ is of Lindbladian type. }\label{fig:GeomAdiabElimLindblad}
\end{figure}
As illustrated on figure~\ref{fig:GeomAdiabElimLindblad}, we  are looking for an expansion based on linear super-operators   $\{K_m\}_{m\geq 0}$ between $\mathcal{D}_s$ and $\mathcal{D}$  and  on Lindblad dynamics $\left\{\fL_{s,m}\right\}_{m\geq 0}$ on $\mathcal{D}_{s}$ such that any  solution $t \mapsto \rho_s(t)\in\mathcal{D}_s$ of the reduced Lindblad master equation
 \begin{equation}\label{eq:DynRed}
    \dotex \rho_s = \fL_s(\rho_s)=\sum_{m\geq 0}   \epsilon^m \fL_{s,m}(\rho_s)
 \end{equation}
yields, via the  map
\begin{equation}\label{eq:KrausRed}
   \rho(t)  = K(\rho_s(t))=\sum_{m\geq 0} \epsilon^m K_m(\rho_s(t)) \; ,
\end{equation}
a trajectory of the perturbed system~\eqref{eq:DynEpsilon}.  We combine here  geometric singular perturbation theory~\cite{Fenichel79} with  center manifold techniques  based on Carr  asymptotic expansion lemma~\cite{carr-book} to derive  recurrence relationships  for  $K_m$ and $\fL_{s,m}$. These recurrences  are obtained by identifying  the terms of the same order in the formal invariance condition
$$
\fL_{0}\big(K(\rho_s)\big)+ \epsilon \fL_{1}\big(K(\rho_s)\big) =  \dotex \rho = K\left(\dotex \rho_s\right) = K \big( \fL_s(\rho_s)\big)
.
$$
This means that, for any $\rho_s\in\mathcal{D}_s$, we have
\begin{multline}\label{eq:invariance}
  \fL_{0}\left( \sum_{m\geq 0} \epsilon^m K_m(\rho_s) \right) + \epsilon \fL_{1}\left( \sum_{m\geq 0} \epsilon^m K_m(\rho_s) \right)
  \\
  = \sum_{m} \epsilon^m K_{m}\left(\sum_{m'} \epsilon^{m'} \fL_{s,m'}(\rho_s) \right)
  .
\end{multline}
For $m=0$ we have
\begin{equation}\label{eq:order0}
  \fL_0\big(K_0(\rho_s)\big) = K_0\big(\fL_{s,0}(\rho_s)\big)
  .
\end{equation}
With $K_0$ defined in~\eqref{eq:K0}, we have $\fL_0\big(K_0(\rho_s)\big) \equiv 0$ and thus $\fL_{s,0}(\rho_s)=0$.
Consequently, for $m\geq 1$, we have
\begin{multline}\label{eq:orderm}
\fL_{0}\left(  K_m(\rho_s) \right) + \fL_{1}\left(  K_{m-1}(\rho_s) \right)
\\
=
\sum_{m'=1}^m K_{m-m'}\left( \fL_{s,m'}(\rho_s) \right)
.
 \end{multline}
 When $m=1$, $K_1$ and $\fL_{s,1}$ are defined by
 \begin{equation}\label{eq:order1}
  \fL_0\big(K_1(\rho_s)\big) + \fL_1\big(K_0(\rho_s)\big) = K_0\big(\fL_{s,1}(\rho_s)\big)
  .
\end{equation}
The following lemma proves that the super-operator $\fL_{s,1}(\rho_s)$ defined by this equation is always of Lindblad form.
\begin{lemma} \label{lem:Ls1}

Assume that  $\fL_1(\rho)= - i[H_1,\rho]$  for some Hermitian operator $H_1$ on $\mathcal{H}$. Then, if $\fL_{s,1}$ satisfies~\eqref{eq:order1}, we have
$\fL_{s,1}(\rho_s) = - i [H_{s,1},\rho_s]$ where $H_{s,1}= S_0^\dag H_1 S_0$  is a Hermitian operator on $\mathcal{H}_s$.

Assume that  $\fL_1(\rho)= L_1 \rho L_1^\dag - \tfrac{1}{2} \big(L_1^\dag L_1 \rho +\rho L_1^\dag L_1  \big)$ for some operator $L_1$ on $\mathcal{H}$. Then, if $\fL_{s,1}$ satisfies~\eqref{eq:order1}, we have
\begin{equation}\label{eq:Ls1}
  \fL_{s,1}(\rho_s)= \sum_\mu A_\mu \rho_s  A_\mu^\dag - \tfrac{1}{2} \big(A_\mu^\dag A_\mu \rho_s +\rho_s A_\mu^\dag A_\mu  \big)
\end{equation}
with $A_\mu= S_0^\dag M_\mu L_1 S_0$ and  the Kraus operators $M_\mu$  defined  by~\eqref{eq:Rkraus}.

The result for a general Lindbladian dynamics \eqref{eq:GenLindblad} for $r=1$ follows by linearity.
\end{lemma}

\begin{proof}
  Since $R\circ \fL_{0} =  0$  and $R\circ K_0=K_0$ , we have $R\left(\fL_1\big(K_0(\rho_s)\big) \right) =K_0\big(\fL_{s,1}(\rho_s)\big)$. Left multiplication by $S_0^\dag$ and right multiplication by $S_0$ yields
  $$
  \fL_{s,1}(\rho_s) = S^\dag_0 R\left(\fL_1\Big(S_0\rho_sS^\dag_0 \Big) \right)  S_0
  $$
  since $S^\dag_0 S_0=I_s$ is the identity operator on $\mathcal{H}_s$.

For  $\fL_1(\rho)= - i[H_1,\rho]$ we have, exploiting the fact that $M_\mu S_0= \lambda_\mu S_0$  and $S_0^\dag S_0=I_s$:
\begin{multline}\label{eq:Trick}
  S^\dag_0  R\left(\fL_1\Big(S_0\rho_sS^\dag_0 \Big)\right)  S_0
  \\
 = - i \sum _\mu S_0^\dag M_\mu \big( H_1 S_0\rho_sS^\dag_0- S_0\rho_sS^\dag_0 H_1\big) M_\mu^\dag S_0
 \\
 = - i \sum _\mu \left(\lambda_\mu^* S_0^\dag \right)M_\mu  H_1 S_0\rho_s + i \sum_\mu  \rho_s S^\dag_0 H_1M_\mu^\dag  \big(\lambda_\mu S^0\big)
 \\
  = - i \sum _\mu  S_0^\dag  M_\mu^\dag M_\mu  H_1 S_0\rho_s + i \sum_\mu  \rho_s S^\dag_0 H_1M_\mu^\dag  M_\mu  S^0
  \\
  = - i \left[  S_0^\dag  H_1  S_0~,~\rho_s \right]
\end{multline}
since $\sum_\mu M_\mu^\dag M_\mu = I$.  We get the Zeno Hamiltonian  $H_{s,1}= S_0^\dag  H_1  S_0$.

For  $\fL_1(\rho)= L_1 \rho L_1^\dag - \tfrac{1}{2} \big(L_1^\dag L_1 \rho +\rho L_1^\dag L_1  \big)$, similar computations yield
\begin{multline*}
  S^\dag_0  R\left(\fL_1\Big(S_0\rho_sS^\dag_0 \Big)\right)  S_0
  \\ = \sum _\mu S_0^\dag M_\mu L_1 S_0\rho_sS^\dag_0 L_1^\dag  M_\mu^\dag S_0
 \\ - \tfrac{1}{2}\sum _\mu S_0^\dag M_\mu  \big(L_1^\dag L_1 S_0\rho_sS^\dag_0 -  S_0\rho_sS^\dag_0 L_1^\dag L_1 \big) M_\mu^\dag S_0
 \\= \left(\sum _\mu A_\mu \rho_s A_\mu^\dag \right) - \tfrac{1}{2}S_0^\dag L_1^\dag L_1 S_0\rho_s-  \rho_sS^\dag_0 L_1^\dag L_1 S_0
\end{multline*}
with $A_\mu=S_0^\dag M_\mu L_1 S_0$.  It remains to prove that $\sum_\mu A_\mu^\dag A_\mu= S_0^\dag L_1^\dag L_1 S_0$ for showing that we indeed have a Lindblad formulation. This results from the following computations:
\begin{multline*}
  \sum_\mu A_\mu^\dag A_\mu= \sum_\mu S_0^\dag  L_1^\dag M_\mu^\dag S_0S_0^\dagger M_\mu L_1 S_0
  \\= S_0^\dag L_1^\dag R^*(S_0 S_0^\dag) L_1 S_0 = S_0^\dag L_1^\dag L_1 S_0,
\end{multline*}
where we use that $S_0 S_0^\dag =P_0$ and $R^*(P_0)=I$.
\end{proof}

%%%  Second order  %%%

\section{Second order expansion for Hamiltonian perturbations} \label{sec:secondorder}

We assume here that $\fL_0$ is defined by a single operator $L_0$,
$
\fL_0(\rho)= L_0 \rho L_0^\dag - \tfrac{1}{2}\big( L_0^\dag L_0 \rho  + \rho L_0^\dag L_0 \big)
,
$
and that the perturbation $\fL_1$ is Hamiltonian,
$
\fL_1(\rho)= - i |H_1,\rho],
$
where  $H_1$ is a Hermitian operator.  The following  lemma gives  a simple expression for $K_1(\rho_s)$ solution of~\eqref{eq:order1}.
\begin{lemma} \label{lem:K1}
Assume that   $\fL_0(\rho)= L_0 \rho L_0^\dag - \tfrac{1}{2}\big( L_0^\dag L_0 \rho  + \rho L_0^\dag L_0 \big) $ and $\fL_1(\rho)=- i[H_1,\rho]$.    Then
$\fL_{s,1}(\rho_s)= -i \big[S_0^\dag H_1 S_0~,~\rho_s\big]$ and $K_1(\rho_s)=- i\big[C_1, S_0 \rho_s S_0^\dag \big] $
satisfy~\eqref{eq:order1} where  $C_1$ is the Hermitian operator
$$
C_1= 2 (L_0^\dag L_0)^{-1} H_1 P_0 +  2  P_0 H_1 (L_0^\dag L_0)^{-1}
$$
with  $P_0$ the orthogonal projector onto $\mathcal{H_0}$ and  $(L_0^\dag L_0)^{-1}$ standing for the Moore-Penrose pseudo-inverse of the Hermitian operator $L_0^\dag L_0$.
\end{lemma}
The associated first order $\rho_s$-parametrization of the slow invariant attractive manifold,
\begin{multline*}
  K_0(\rho_s)+ \epsilon K_1(\rho_s)
=
\\
\left( I - i\epsilon (L_0^\dag L_0)^{-1} H_1\right) S_0 \rho_s S_0^\dag  \left( I + i\epsilon (L_0^\dag L_0)^{-1} H_1\right)
+ 0(\epsilon^2)
,
\end{multline*}
corresponds, up to second-order terms, to a trace-preserving completely positive map.

\begin{proof}
With  $S_0 \fL_{s,1}(\rho_s) S_0^\dag = -i\big[P_0 H_1 P_0,S_0\rho_s S_0^\dag\big]$, \eqref{eq:order1}  reads
\begin{multline*}
  \fL_0(K_1(\rho_s))=-i\big[P_0 H_1 P_0,S_0\rho_s S_0^\dag\big] + i\big[H_1,S_0\rho_s S_0^\dag\big]
\\
= -i \left[ P_0 H_1 P_0 -H_1~, ~S_0\rho_s S_0^\dag\right]
.
\end{multline*}
With $K_1(\rho_s)= -i\big[C_1, S_0 \rho_s S_0^\dag \big] $  we have  also
\begin{multline*}
 \fL_0(K_1(\rho_s)) = - i L_0 \left[ C_1, S_0\rho_s S_0^\dag\right] L_0^\dag
 \\
 + \tfrac{i}{2}\left( L_0^\dag L_0 \left[ C_1, S_0\rho_s S_0^\dag\right]  +  \left[ C_1, S_0\rho_s S_0^\dag\right]  L_0^\dag L_0 \right)
 .
\end{multline*}
Since $L_0 S_0=0$ and $S_0^\dag L_0^\dag=0$  we have  $$L_0 \left[ C_1, S_0\rho_s S_0^\dag\right] L_0^\dag =0.$$
Since  additionally, $P_0 S_0=S_0$,    $L_0^\dag L_0 P_0=0$ and  $L_0^\dag L_0 (L_0^\dag L_0)^{-1} = I-P_0$, we have
$$
L_0^\dag L_0 \left[ C_1, S_0\rho_s S_0^\dag\right] = 2 (I-P_0) H_1 P_0 S_0 \rho_s S_0^\dag
.
$$
Thus
\begin{multline*}
 \fL_0(K_1(\rho_s)) =   i (I-P_0) H_1 P_0 S_0 \rho_s S_0^\dag -  i S_0 \rho_s S_0^\dag   P_0 H_1 (I-P_0)
  \\
   =
- i \left[ P_0 H_1 P_0 -H_1~, ~S_0\rho_s S_0^\dag\right]
 .
\end{multline*}
\end{proof}

The second order term  $\fL_{s,2}(\rho_s)$ is  solution of~\eqref{eq:orderm} for $m=2$:
$$
  \fL_{0}\left(  K_2(\rho_s) \right) + \fL_{1}\left(  K_{1}(\rho_s) \right)
=
K_{0}\left( \fL_{s,2}(\rho_s) \right) + K_{1}\left( \fL_{s,1}(\rho_s) \right)
.
$$
Using, once again, $R\circ \fL_0\equiv 0$ and $R \circ K_0 = K_0$, we get
\begin{equation}\label{eq:Ls2}
\fL_{s,2}(\rho_s)
=
S_0^\dag  R\Big(\fL_{1}\left(  K_{1}(\rho_s) \right) -  K_{1}\left( \fL_{s,1}(\rho_s) \right)\Big) S_0
.
\end{equation}
The following lemma shows that  $\fL_{s,2}(\rho_s)$ admits a Lindbladian form.

\begin{lemma}\label{lem:Fs2}
 The super-operator $\fL_{s,2}$ defined by~\eqref{eq:Ls2} admits the following Lindbladian formulation
$$
  \fL_{s,2}(\rho_s)= \sum_\mu B_\mu \rho_s  B_\mu^\dag - \tfrac{1}{2} \big(B_\mu^\dag B_\mu \rho_s +\rho_s B_\mu^\dag B_\mu  \big)
$$
with $B_\mu= 2 S_0^\dag M_\mu  L_0 (L_0^\dag L_0)^{-1} H_1 S_0$, $M_\mu$  defined  by~\eqref{eq:Rkraus} and $(L_0^\dag L_0)^{-1}$ standing for  the Moore-Penrose pseudo-inverse of $L_0^\dag L_0$.
\end{lemma}

\begin{proof}
 We have  $ R\Big(  K_{1}\left( \fL_{s,1}(\rho_s) \right)\Big) =0$. This results from ($\rkz$ stands for $S_0 \rho_s S_0^\dag = K_0(\rho_s)$)
 \begin{multline}
  K_{1}\left( \fL_{s,1}(\rho_s) \right)
  = -i[C_1,\, -i\,S_0[S_0^\dag H_1 S_0,\,\rho_s]S_0^\dag] \\
  = -[C_1, P_0H_1 \rkz - \rkz H_1 P_0]
   \\
   \label{eq:Trick2}
  = -2 (L_0^\dag L_0)^{-1} H_1 (P_0H_1\rkz- \rkz H_1P_0) \\
    +2 (P_0H_1\rkz- \rkz H_1P_0)  H_1 (L_0^\dag L_0)^{-1}
\end{multline}
where we have used Lemma \ref{lem:K1} and $P_0 K_0 = K_0$.

Repeating computations similar to \eqref{eq:Trick}, we see that for any operator $A$ on $\mathcal{H}$, $R(AP_0)=R(P_0A)=P_0 A P_0$. Since $P_0 \rkz= \rkz P_0=\rkz$ we moreover have $R(A\rkz)=P_0 A \rkz$ and $R(\rkz A)=\rkz A P_0$. This gives the result of applying $R$ on all the terms in \eqref{eq:Trick2}, and since $P_0(L_0^\dag L_0)^{-1} = (L_0^\dag L_0)^{-1}P_0=0$, we conclude that $R(K_1(\fL_{s,1}))= 0$.

Thus $\fL_{s,2}(\rho_s) = S_0^\dag  R\Big(\fL_{1}\left(  K_{1}(\rho_s) \right) \Big) S_0$. Exploiting similar  simplifications, we have
\begin{multline*}
  \fL_{1}\left(  K_{1}(\rho_s) \right) = -H_1(C_1\rkz-\rkz C_1) +(C_1\rkz-\rkz C_1) H_1
  \\
  = H_1 \rkz C_1 + C_1 \rkz H_1 - (H_1 C_1 \rkz+\rkz C_1 H_1)
  \\
  = 2H_1 \rkz H_1(L_0^\dag L_0)^{-1} + 2 (L_0^\dag L_0)^{-1} H_1 \rkz H_1
  \\ - 2 H_1 (L_0^\dag L_0)^{-1} H_1 \rkz -2  \rkz H_1 (L_0^\dag L_0)^{-1} H_1
\end{multline*}
and, using $S_0^\dag R(A\rkz)=S_0^\dag P_0 A \rkz = S_0^\dag A \rkz$ and the definition $\rkz = S_0 \rho_s s_0^\dag$, we get
\begin{multline*}
 \fL_{s,2}(\rho_s)
  =
  \\2 S_0^\dag R \Big( H_1 \rkz H_1(L_0^\dag L_0)^{-1} + (L_0^\dag L_0)^{-1} H_1 \rkz H_1\Big) S_0
  \\ - 2 S^\dag_0H_1 (L_0^\dag L_0)^{-1} H_1 S_0 \rho_s -2 \rho_s S^\dag_0H_1 (L_0^\dag L_0)^{-1} H_1 S_0
  .
\end{multline*}
Since for all $A$, $R(\fL_0(A))=0$, we have the identity
$$
R(L_0 A L_0^\dag) = R \big( \tfrac{1}{2} \big(L_0^\dag L_0 A + A L_0^\dag L_0\big) \big)
.
$$
With $A=(L_0^\dag L_0)^{-1}  H_1 \rkz H_1 (L_0^\dag L_0)^{-1}$ we get
\begin{multline*}
2 R\Big(L_0 (L_0^\dag L_0)^{-1}  H_1 \rkz H_1 (L_0^\dag L_0)^{-1} L_0^\dag\Big)  =
  \\
 R\Big( (I-P_0) H_1 \rkz H_1 (L_0^\dag L_0)^{-1} +  (L_0^\dag L_0)^{-1}  H_1 \rkz H_1 (I-P_0)\Big)
 \\
  = R\Big(  H_1 \rkz H_1 (L_0^\dag L_0)^{-1} +  (L_0^\dag L_0)^{-1}  H_1 \rkz H_1 \Big)
\end{multline*}
since $
R\Big( P_0 H_1 \rkz H_1 (L_0^\dag L_0)^{-1} \Big)=P_0 H_1 \rkz H_1 (L_0^\dag L_0)^{-1}P_0$
 and  $(L_0^\dag L_0)^{-1}P_0=0$.
Thus
\begin{multline*}
 \fL_{s,2}(\rho_s)
  =
  \\4 S_0^\dag  R\Big(L_0 (L_0^\dag L_0)^{-1}  H_1 \rkz H_1 (L_0^\dag L_0)^{-1} L_0^\dag\Big) S_0
  \\ - 2 S^\dag_0 H_1 (L_0^\dag L_0)^{-1} H_1 S_0 \rho_s -2 \rho_s S^\dag_0H_1 (L_0^\dag L_0)^{-1} H_1 S_0
  .
\end{multline*}
Using the decomposition~\eqref{eq:Rkraus} of $R$ we  have
$$
4 S_0^\dag  R\Big(L_0 (L_0^\dag L_0)^{-1}  H_1 \rkz H_1 (L_0^\dag L_0)^{-1} L_0^\dag\Big) S_0 \\
= \sum_\mu B_\mu \rho_s B_\mu^\dag \, .
$$
We conclude by the following computations:
\begin{multline*}
 \tfrac{1}{2}\sum_\mu B_\mu^\dag B_\mu=
 \\ 2 \sum_\mu  S_0^\dag H_1(L_0^\dag L_0)^{-1}L_0^\dag  M_\mu^\dag S_0 S_0^\dag M_\mu  L_0 (L_0^\dag L_0)^{-1} H_1 S_0
 \\
 = 2   S_0^\dag H_1(L_0^\dag L_0)^{-1}L_0^\dag  R^*(P_0)  L_0 (L_0^\dag L_0)^{-1} H_1 S_0
 \\
 = 2   S_0^\dag H_1(L_0^\dag L_0)^{-1}L_0^\dag   L_0 (L_0^\dag L_0)^{-1} H_1 S_0
 \\
 = 2   S_0^\dag H_1 (L_0^\dag L_0)^{-1} H_1 S_0
 .
\end{multline*}
\end{proof}

%%%  SECTION:  EXAMPLE   %%%

\section{Illustrative example: low-Q cavity coupled to another quantum system} \label{sec:example}

The developments above are rigorous in finite dimension, but they can be formally applied also on infinite- dimensional systems, as illsutrated in the following example.

We consider a strongly dissipative driven harmonic oscillator (low-Q cavity) coupled to another, undamped quantum system with the same transition frequency (``target'' system). Denote $\mathcal{H}_A$ (resp. $\mathcal{H}_B$) the infinite-dimensional Hilbert space of the strongly dissipative harmonic oscillator (resp. the target system), spanned by the Fock states $\{\ket{n_A}\}_{n\in \mathbb{N}}$ (resp.~a possibly infinite basis $\{\ket{n_B}\}_{n_B}$); $\rho$ is the density operator of the composite system, on $\mathcal{H}=\mathcal{H}_A \otimes \mathcal{H}_B$.

In the frame rotating at the common frequency of the two systems, their coupled evolution is described by the standard master differential equation:
\begin{align}\label{eq:ExDyn1}
\begin{aligned}
\frac{d}{dt}\rho = [u\tilde\ba^\dagger - u^*\tilde\ba, \rho] + \kappa \left(\tilde\ba\rho \tilde\ba^\dagger-\frac{1}{2}\left(\tilde\ba^\dagger \tilde\ba \rho + \rho \tilde\ba^\dagger \tilde\ba \right) \right) \\ - ig\left[\tilde\ba^\dagger \tilde\bb + \tilde\ba \tilde\bb^\dagger , \rho \right] \, .\\
\end{aligned}
\end{align}
Here $\tilde\ba=\ba \otimes I_B $ and $\tilde\bb = I_A \otimes \bb$ are the annihilation operators respectively for the harmonic oscillator $A$ and for the quantum system $B$ (possibly generalized if $B$ is not a harmonic oscillator; e.g.~if $B$ is a qubit, we have $\bb = \ket{g}\bra{e}$ the transition operator from excited to ground state). The first line describes the driven and damped evolution of harmonic oscillator A, while the second line describes the exchange of energy quanta between the two quantum systems.
The constants $(\kappa, g) \in \mathbb{R}^2$ govern the speed of these dynamics. We here consider $\kappa \gg g$, with the goal to adiabatically eliminate the fast dynamics of the low-Q cavity and compute its effect on the other quantum system. The dynamics \eqref{eq:ExDyn1} is then equivalent to
\begin{align}\label{eq:mast_eq_ex}
\frac{d}{dt}\rho = \mathfrak{L}_0(\rho) + \epsilon\mathfrak{L}_1(\rho)
\end{align}
with $L_0 = \sqrt{\kappa}(\tilde\ba - \alpha)$, $\alpha = 2u/\kappa$ and $\epsilon\mathfrak{L}_1(\rho)=- ig\left[\tilde\ba^\dagger \tilde\bb + \tilde\ba \tilde\bb^\dagger , \rho \right]$. For this typical example, the results of $\mathfrak{L}_{s,1}$ and $\mathfrak{L}_{s,2}$ are well known (see e.g. \cite[chap.12]{CarmichaelBook07}). Our results allow to readily retrieve their expression and thus completely circumvent the trouble of the usual calculation.

In the absence of coupling between the two subsystems ($\epsilon=0$), the overall system trivially converges towards $R(\rho_0) = \ket{\alpha}\bra{\alpha}_A \otimes \text{Tr}_A(\rho(0))$. Here $\text{Tr}_A$ is the partial trace over $\mathcal{H}_A$ and $\ket{\alpha}$ denotes the coherent state of amplitude $\alpha \in \mathbb{C}$, towards which a classically driven and damped harmonic oscillator is known to converge. Therefore we have $\mathcal{H}_0 = \ket{\alpha}\bra{\alpha} \otimes \mathcal{H}_B$, $P_0 = \ket{\alpha}\bra{\alpha} \otimes I_B$, and $M_\mu = \ket{\alpha}\bra{\mu_A} \otimes I_B$ with $\mu$ spanning $\mathbb{N}$. We will naturally describe $\rho_S$ on the Hilbert space $\mathcal{H}_s \equiv \mathcal{H}_B$ and with basis $\{\ket{n_s}\}_{n_s}$, so $S_0 = \sum_{n}\ket{\alpha}\ket{n_B}\bra{n_s}$.

For the first-order perturbation, using the property $\tilde\ba \ket{\alpha} = \alpha \ket{\alpha}$, Lemma \ref{lem:Ls1} readily yields
$$H_{s,1} = \alpha \bb_s^\dagger + \alpha^* \bb_s \; ,$$
denoting by $\boldsymbol{q}_s$ the operator on $\mathcal{H}_s$ equivalent to $\boldsymbol{q}$ on $\mathcal{H}_B$. This standard result shows that the oscillator $A$ can be approximated as a classical field of amplitude $\alpha$. Indeed, $H_{s,1}$ describes e.g.~Rabi oscillations for a qubit driven by a classical field ($\mathcal{H}_B=\text{span}\{\ket{g},\ket{e}\}$); or, when $\mathcal{H}_B$ describes another harmonic oscillator, $H_{s,1}$ is the same Hamiltonian in fact as in the first line of \eqref{eq:ExDyn1}, with classical drive amplitude $i u$ replaced by $\alpha$.

Next, using $\Dp$ the unitary displacement operator on $\mathcal{H}_A$, which satisfies $\Dp \ba \Dm = \ba-\alpha I$, we compute $(L_0^\dag L_0)^{-1} = \Dp \bN_A^{-1} \Dm/\kappa$, where $\bN_A = \ba^\dag \ba = \sum_{n \in \mathbb{N}} \; n \, \ket{n_A}\bra{n_A}$ and the Moore-Penrose pseudo-inverse of $\bN_A$ is just $\bN_A^{-1}=\sum_{n \geq 1} \; \frac{1}{n} \, \ket{n_A}\bra{n_A}$. We then compute
%AS: couper des trucs ici?
\begin{multline*}
C_1 = \frac{2}{\kappa}\, \Dp \bN_A^{-1} \Dm (\tilde\ba^\dag \tilde\bb + \tilde\ba \tilde\bb^\dag) \ket{\alpha}\bra{\alpha}\otimes I_B \;\; + h.c. \\
= \frac{2}{\kappa}\, \Dp \bN_A^{-1} ((\tilde\ba^\dag+\alpha^* I) \tilde\bb + (\tilde\ba+\alpha I) \tilde\bb^\dag) \Dm \ket{\alpha}\bra{\alpha}\otimes I_B \;\; + h.c. \\
= \frac{2}{\kappa}\, \Dp \bN_A^{-1} ((\tilde\ba^\dag+\alpha^* I) \tilde\bb + (\tilde\ba+\alpha I) \tilde\bb^\dag) \ket{0}\bra{\alpha}\otimes I_B \;\; + h.c. \\
= \frac{2}{\kappa}\, \Dp \bN_A^{-1} ((\alpha^*\tilde\bb + \alpha \tilde\bb^\dag)\ket{0} + \tilde\bb \ket{1})\bra{\alpha} \;\; + h.c. \\
= \frac{2}{\kappa}\, \Dp \ket{1}\bra{\alpha}\otimes \bb \;\; + h.c. \, .
\end{multline*}
From Lemma \ref{lem:K1}, we see that a pure state $\ket{\psi_S} \in \mathcal{H}_s$ gets mapped at order zero to $\ket{\alpha} \otimes \ket{\psi_B}$ with $\ket{\psi_B} \equiv \ket{\psi_S}$, but at order one to a slightly rotated state $\ket{\alpha} \otimes \ket{\psi_B} - \frac{i g}{\kappa} (\Dp \ket{1}) \otimes (\bb \ket{\psi_B})$. This expresses that the coupled low-Q cavity $A$ contains slightly more energy than a coherent state, to the detriment of system B.

For the second order perturbation, from Lemma \ref{lem:Fs2} we must compute $B_\mu= 2 S_0^\dag M_\mu  L_0 (L_0^\dag L_0)^{-1} H_1 S_0$. The computations made for $C_1$ above can be used, writing:
\begin{multline*}
B_\mu = S_0^\dag M_\mu  L_0 \; \big( \frac{2}{\kappa}\, \Dp \ket{1}\bra{\alpha}\otimes \bb \big) S_0 \\
= \frac{2}{\sqrt{\kappa}} S_0^\dag M_\mu \big( \Dp \ba \ket{1}\bra{\alpha}\otimes \bb \big) \, S_0 \\
= \frac{2}{\sqrt{\kappa}} \sum_{n,m} \ket{n_s}\bra{n_B} \bra{\mu_A} \; \ket{\alpha}\bra{\alpha}\otimes \bb \; \ket{m_B}\bra{m_S}\\
= \frac{2}{\sqrt{\kappa}}\; \braket{\mu_A}{\alpha}\bb_s \, .
\end{multline*}
All the obtained $B_\mu$ are in fact identical up to a scalar factor, so they may be combined into a single operator:
\begin{multline*}
\epsilon^2  \fL_{s,2}(\rho_s)=g^2 \sum_\mu B_\mu \rho_s  B_\mu^\dag - \tfrac{1}{2} \big(B_\mu^\dag B_\mu \rho_s +\rho_s B_\mu^\dag B_\mu  \big) \\
= \frac{4g^2}{\kappa} \sum_\mu |\braket{\mu_A}{\alpha}|^2 \left( \bb_s\rho_s \bb_s^\dagger - \frac{1}{2}\left(\bb_s^\dag \bb_s \rho_s + \rho_s \bb_s^\dagger \bb_s \right) \right) \\
= \frac{4g^2}{\kappa}\left( b\rho_s b^\dagger - \frac{1}{2}\left(b^\dagger b \rho_s + \rho_s b^\dagger b \right)\right) \, .
\end{multline*}
%AS: take out?
(Note that $\{ |\braket{\mu_A}{\alpha}|^2 \}_{\mu \in \mathbb{N}}$ just corresponds to the expansion of the coherent state $\ket{\alpha}$, of unit norm, in the Fock basis.)
We thus get the expected reduced dynamics:
\begin{multline*}
\dotex \rho_s = -ig\left[\alpha \bb_s^\dag + \alpha^*\bb_s, \rho_s \right] \\+ \frac{4g^2}{\kappa}\left( \bb_s\rho_s \bb_s^\dagger - \frac{1}{2}\left(\bb_s^\dagger \bb_s \rho_s + \rho_s \bb_s^\dagger \bb_s \right)\right) \; ,
\end{multline*}
which expresses that the B system is subject to slow damping due to the presence of the low-Q cavity.

\paragraph*{Remark} Note that if the slow dynamics includes a Hamiltonian that acts only on the $B$ system, i.e.~of the form $\tilde\bH_B=I_A \otimes \bH_B$ (acting only on B), then $C_1$ features an additional term
\begin{multline*}
\frac{2}{\kappa}\, \Dp \bN_A^{-1} \Dm \big(I_A \otimes \bH_B\big) \ket{\alpha}\bra{\alpha}\otimes I_B \;\; + h.c.\\
= \frac{2}{\kappa}\, \big(\Dp \bN_A^{-1} \Dm \ket{\alpha}\bra{\alpha} \big) \otimes \bH_B \\
= \frac{2}{\kappa}\, \big(\Dp \bN_A^{-1} \ket{0}\bra{\alpha} \big) \otimes \bH_B = 0 \, .
\end{multline*}
Thus the second-order correction vanishes and the Zeno dynamics is the only addition up to second order:
\begin{multline*}
\dotex \rho_s = -ig\left[\alpha \bb_s^\dag + \alpha^*\bb_s + \bH_B, \rho_s \right] \\+ \frac{4g^2}{\kappa}\left( \bb_s\rho_s \bb_s^\dagger - \frac{1}{2}\left(\bb_s^\dagger \bb_s \rho_s + \rho_s \bb_s^\dagger \bb_s \right)\right) \; .
\end{multline*}

\section{Conclusion}

We have shown how to eliminate the fast dynamics in an open quantum system (Lindblad equation) with two timescales, and obtain the resulting reduced dynamics explicitly in Lindblad form. This is important to guarantee that the approximate model preserves the structure of quantum states (positivity and  trace). The slow system is hence parameterized explicitly with a quantum state on a lower-dimensional Hilbert space, and mapped to the complete Hilbert space by a completely positive trace preserving map (Kraus map). We have illustrated on a benchmark system (highly dissipating quantum oscillator resonantly coupled to another quantum system) how our explicit formulae directly retrieve the results previously obtained with lengthy ad hoc computations.

We have obtained explicit formulae for the second-order corrections only in the particular case of a fast Lindbladian with single-channel damping $L_0$, and a slow ``perturbation'' in Hamiltonian form. Conceptually there should be no obstacle to extending this theory to any Lindbladians, the key point being an appropriate way to generalize the pseudo-inversion $(L_0^\dag L_0)^{-1}$. However, the special case completed here will already allow to answer currently open questions about the influence of small Hamiltonian perturbations on stable open quantum systems built e.g.~with engineered reservoirs \cite{MirrahimiCatComp2014,LeghtTPKVPSNSHRFSMD2015S}.

\section*{Acknowledgement}
The authors thank Benjamin Huard, Zaki Leghtas and Mazyar Mirrahimi for many useful discussions.

\bibliographystyle{plain}
%\bibliography{C:/Users/Remi/Dropbox/These/bib}
%\bibliography{RouchonJabRef}
%\bibliography{E:/Latex/RouchonJabRef}

\end{document}